\documentclass[aps,pre,reprint,twocolumn,floatfix]{revtex4-2} % linenumbers ,endfloats,floatfix,titlepage,

\usepackage[utf8]{inputenc}

\usepackage{graphicx,amsmath,amssymb,color}
\usepackage{setspace}
\usepackage{nima}
\usepackage{dsfont,bbm}
\usepackage[title]{appendix}

\DeclareMathOperator*{\argmin}{arg\,min}

\newtheorem{theorem}{Theorem}

\newenvironment{proof}{\textit{Proof: }}{\hfill$\square$}

\usepackage{theoremref}

\begin{document}

\title{Approximate Network Symmetry}
\author{Yanchen Liu}
\email{To whom correspondence should be addressed. Email: lycrdfzpku@gmail.com}

\affiliation{Center for Complex Network Research, Northeastern University, Boston, USA}

\date{\today}

% \env{abstract}

\begin{abstract}
    \scriptsize
    We define a new measure of network symmetry that is capable of capturing approximate global symmetries of networks, extending previous literature on network symmetry that focuses on mainly perfect network symmetry. We apply this measure to different networks sampled from several classic network models, as well as several real-world networks. We find that among the network models that we have examined, Erdös-Rényi networks have the least levels of symmetry, and Random Geometric Graphs are likely to have high levels of symmetry. We find that our network symmetry measure can capture properties of network structure, and help us gain insights on the structure of real-world networks. Moreover, our network symmetry measure is capable of capturing imperfect network symmetry, which would have been undetected if only perfect symmetry is considered.
    % , as in most network symmetry literature. 
\end{abstract}

\maketitle

\section{Introduction}

Recent work has revealed that complex networks can have a variety of different symmetries \cite{macarthur2006symmetry,macarthur2008symmetry,ball2018symmetric,xiao2008emergence,macarthur2006symmetry,macarthur2008symmetry,morone2020fibration,leifer2020circuits,xiao2008network}. 
Network symmetry has been shown to be useful in many complex network research areas. For example, in the studies of network dynamics, it has been shown that there is a direct connection between network symmetry and the formation of synchronized node clusters \cite{pecora2014cluster, cho2017stable, sorrentino2016complete}. Also, network symmetry has been used in reducing networks to their structural skeletons, also known as network quotients \cite{xiao2008network, boldi2002fibrations}. 
This technique is useful in studying dynsamics on protein-protein interaction networks \cite{rietman2011review}, and gene regulatory networks \cite{morone2020fibration,leifer2020circuits}.
Network symmetry is also related to the adjacency matrix spectrum, causing additional degeneracies of certain eigenvalues \cite{macarthur2009spectral,lauri2016topics}.
It has also been observed that many real-world networks have surprisingly high levels of symmetry \cite{ball2018symmetric,xiao2008emergence,macarthur2006symmetry,macarthur2008symmetry}.

Most papers on network symmetry use graph automorphism groups to evaluate the symmetries of networks \cite{macarthur2006symmetry,macarthur2008symmetry}. An automorphism of a network is a node permutation $\mathbbm{P}: i\rightarrow \mathbbm{P}_i$ that preserves the adjacency matrix $A$ of the network: $A_{ij} = A_{\mathbbm{P}_i \mathbbm{P}_j}$ for $i,j \in \{1,..,N\}$ \cite{as1998modern}. Specifically, if there is an edge between nodes $i,j$, then an automorphic permutation $\mathbbm{P}$ yields that there is an edge between nodes $\mathbbm{P}_i, \mathbbm{P}_j$ (Figure~\ref{fig:permutation_exmp}). 
% preserves the adjacency matrix $A$, then given $A_{ij} = 1$, we have $A_{\mathbbm{P}_i \mathbbm{P}_j} = 1$.  
Figure~\ref{fig:symmetric_net_exmps} shows three examples of symmetric networks, and examples of their automorphism $\mathbbm{P}$.
A permutation of nodes can be manifested in the adjacency matrix as a permutation of the rows and columns, and a permutation is automorphic if the adjacency matrix after the permutation is exactly the same as the original one. 
A permutation $\mathbbm{P}$ of $N$ nodes can be encoded in a permutation matrix $P$, which is defined such that $P_{ij}=1$ if the $j$th-ordered node in the original network is moved to $i$th slot in the new node order, otherwise $P_{ij}=0$ (Figure~\ref{fig:permutation_exmp}). 
Such permutation matrices $P$ have the properties that each row and column has only one non-zero element, which is equal to $1$ ($\sum_{j=1}^{N} P_{ij} = 1$ for $i\in \{1,..,N\}$ and $\sum_{i=1}^{N} P_{ij} = 1$ for $j\in \{1,..,N\}$), and $PP^T = I$. 
An automorphism ($\mathbbm{P}$) of a network with adjacency matrix $A$ is equivalent to $A = PAP^T$, 
% equivalent to $AP-PA = 0$. In other words, if $P$ encodes an automorphism, $P$ and $A$ commutes.
where $PAP^T$ is the adjacency matrix after permutation.
In most literature on network symmetry, a network is defined to be {\it symmetric} if there exists a non-trivial automorphism ($P\neq I$, $I$ is the identity matrix) of the network, and asymmetric if there is no non-trivial automorphism (Figure~\ref{fig:4x4_2D_Lattices}) \cite{macarthur2006symmetry,macarthur2008symmetry}. 
The size of the automorphism group of a network, or the number of nodes involved in the automorphism group, have been used as measures to quantify the level of symmetry \cite{bollobas1979graph,macarthur2006symmetry,macarthur2008symmetry,macarthur2007automorphism}.

\begin{figure}
    \centering
    \includegraphics[width=0.48\textwidth]{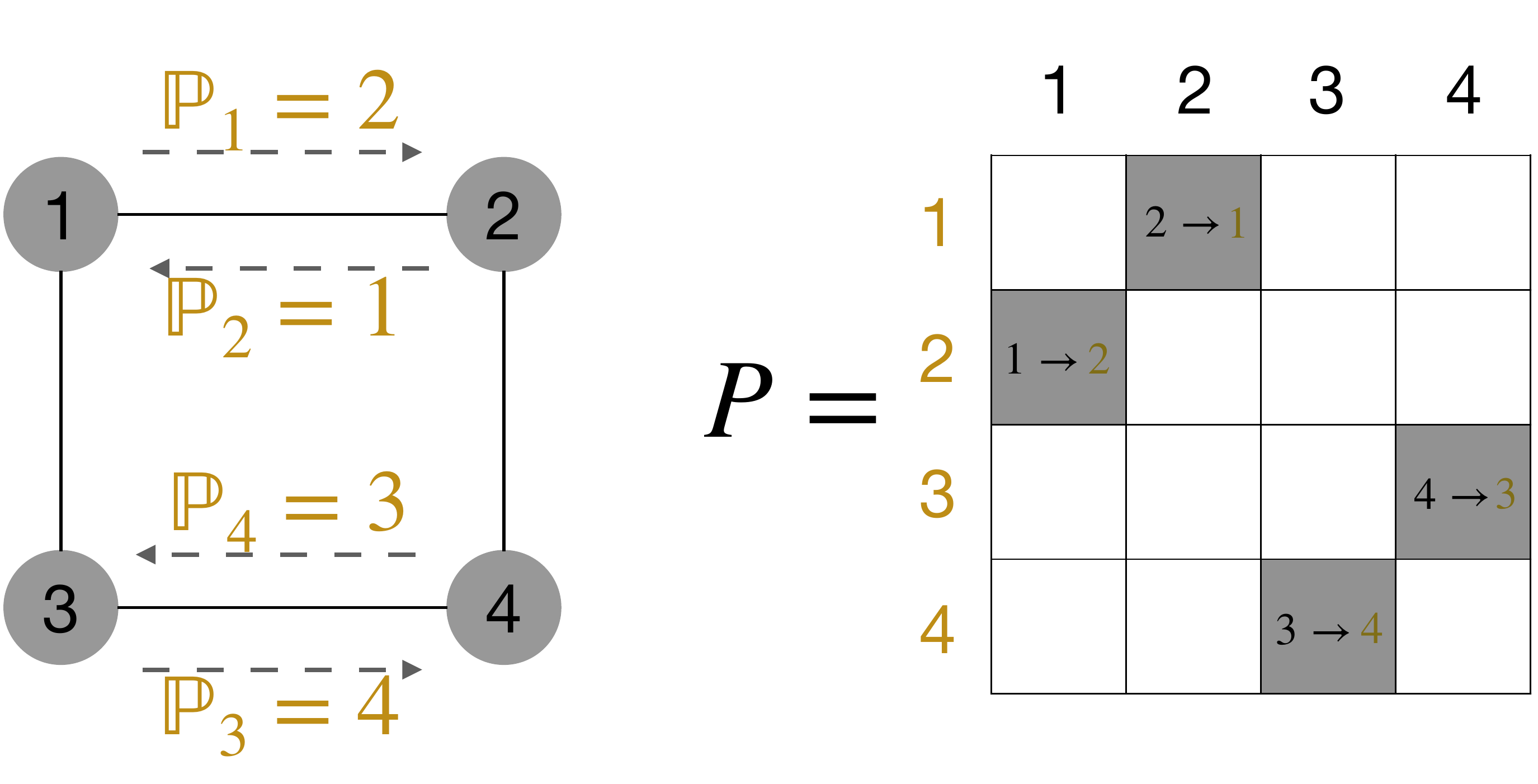}
    \caption{\scriptsize
    A four-node square network, and one of its automorphisms node permutation $\mathbbm{P}$. The node permutation matrix $P$ has non-zero elements $P_{12}$, $P_{21}$, $P_{34}$ and $P_{43}$. A non-zero element $P_{ij}$ represents a permutation of node $i$ to node $j$'s slot.
    The columns' labels of $P$ corresponds to the nodes' old labels, and the rows' labels of $P$ corresponds to  the nodes' new labels.
    }
    \label{fig:permutation_exmp}
\end{figure}

\begin{figure}
    \centering
    \includegraphics[width=0.46\textwidth]{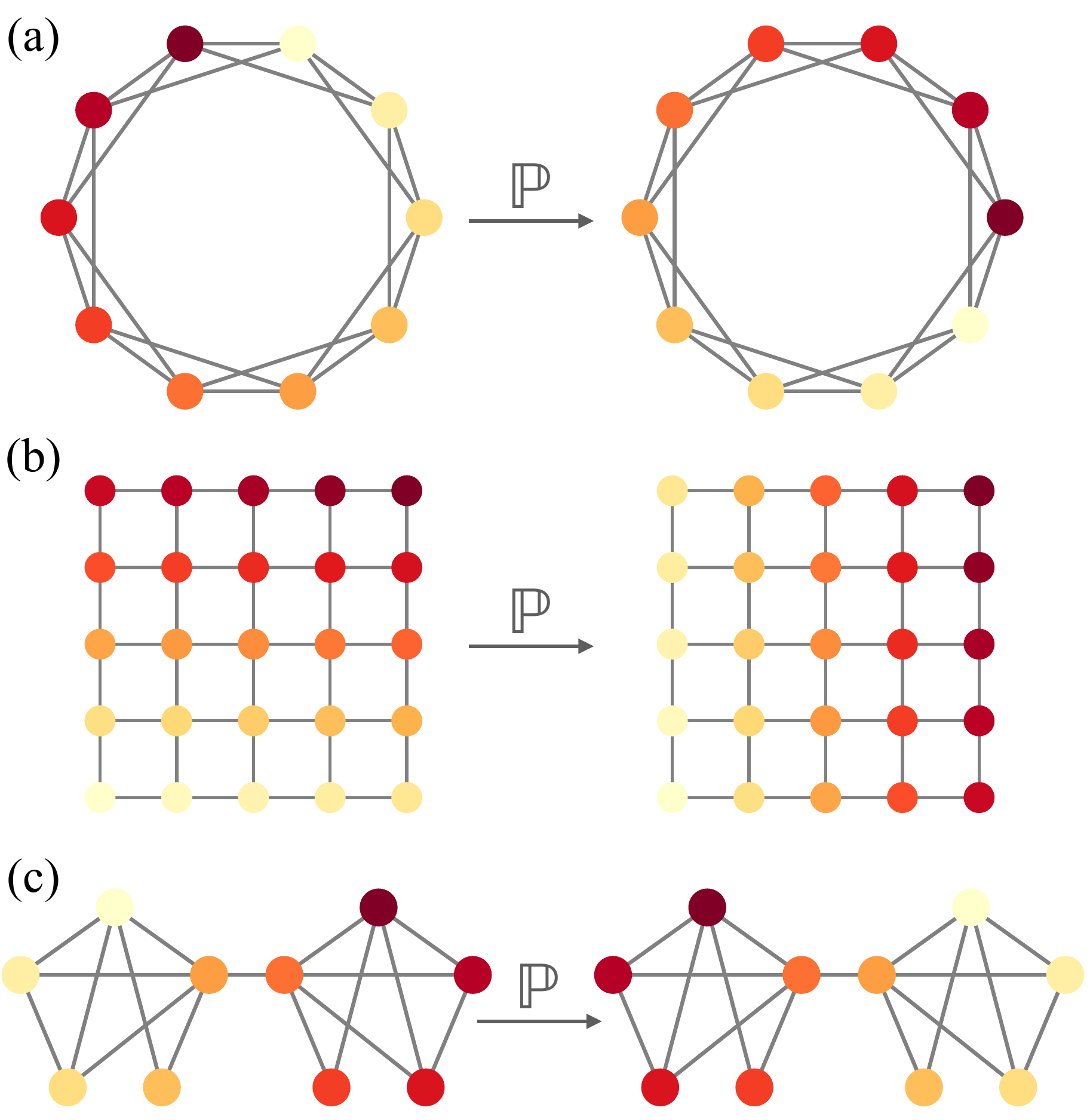}
    \caption{\scriptsize
    Three examples of networks with perfect symmetry.
    (a): A ring lattice network with rotational symmetry.
    (b): A 2D lattice network with $\pi/2$-rotational and mirror symmetry.
    (c): A network with mirror symmetry.
    }
    \label{fig:small_2D_Lattice_examples}
    \label{fig:symmetric_net_exmps}
    \label{fig:WS_N50_k10_before_after_perm}
\end{figure}

\begin{figure}
    \centering
    \includegraphics[width=0.4\textwidth]{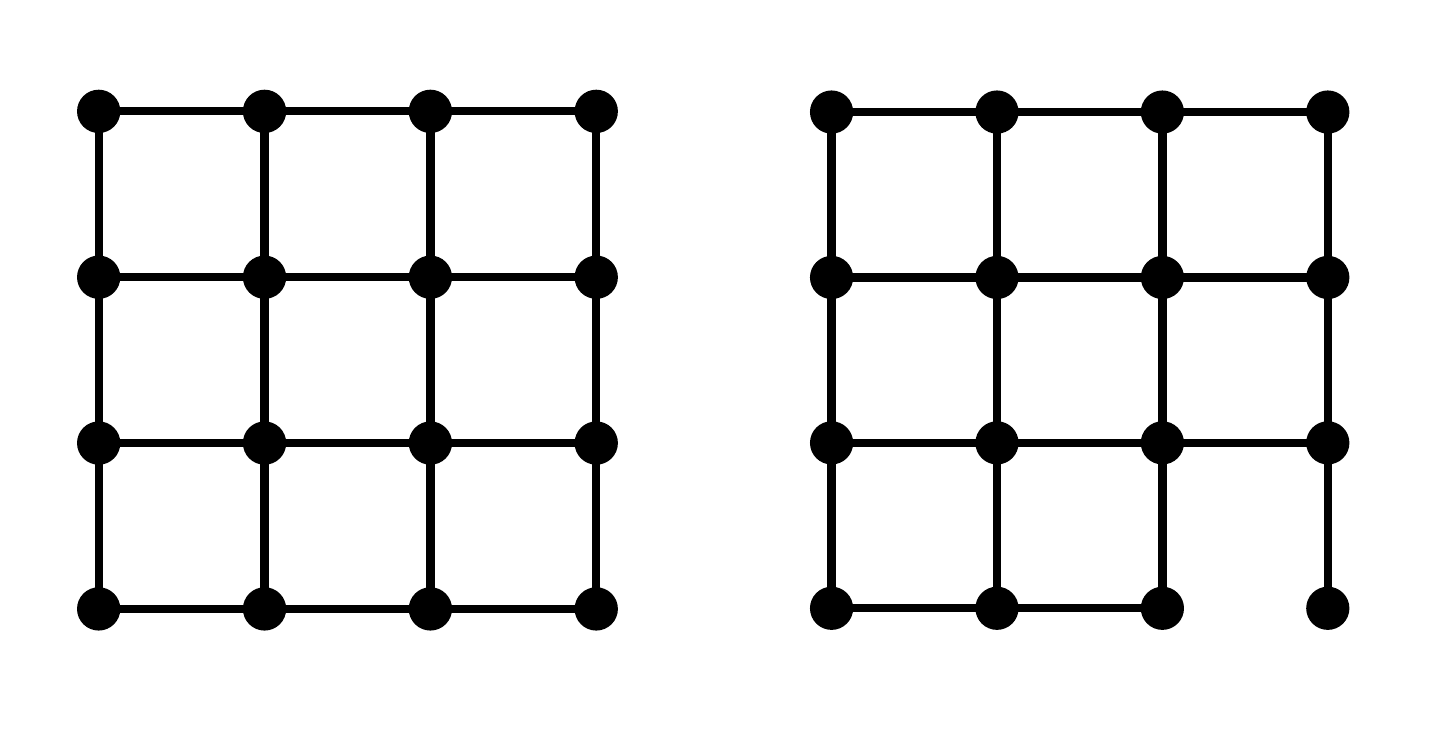}
    \caption{\scriptsize A $4$ by $4$ 2D lattice network, and the same lattice network with one edge missing. 
    The lattice network on the left is symmetric under the definition of network symmetry via automorphism \cite{macarthur2006symmetry,macarthur2008symmetry}, because it has many automorphic permutations, for example, node `reflections' along the diagonals or node `rotations' for multiples of $\pi/2$. 
    The network on the right, however, has no nontrivial automorphism, which means that it is evaluated as having no symmetry by most commonly used symmetry measures via automorphism. 
    This is an example of the fragility of network automorphism in some cases, especially in the cases of global symmetries. 
    Intuitively, the network on the right still exhibits a certain non-zero level of symmetry, albeit being less than that of the network on the left.
    }
    \label{fig:4x4_2D_Lattices}
\end{figure}

The level of symmetry of many real-world complex networks has been analyzed via their automorphism groups, and it is found that most of them have nontrivial symmetries \cite{macarthur2007automorphism,xiao2008emergence}. This is surprising considering that under the same definition of symmetry, almost all large random graphs (for example Erdös-Rényi random graphs) are asymmetric \cite{bollobas2001random}.
In \cite{macarthur2007automorphism} the authors analyzed several real-world networks, including several genetic interaction networks, the internet network, and the US power grid network. They found that almost all symmetries in these real-world networks can be decomposed into clique or biclique symmetry (Figure~\ref{fig:local_syms}).
They explained that the existence of locally tree-like structures in the real-world networks contributes mainly to their high levels of symmetry, since random trees exhibit high levels of symmetry (derived from identical branches originated from the same root-node, for example Figure~\ref{fig:local_syms}(a,b)) \cite{harary1979probability}.
% It is believed that the growth process through which many real-world networks are formed naturally leads to locally tree-like regions in the networks \cite{macarthur2007automorphism}, and since random trees exhibit high levels of symmetry (derived from identical branches originated from the same root-node) \cite{harary1979probability}, it is conceivable that the high level of symmetry in most real-world networks comes from their locally tree-like structures (Figure~\ref{fig:local_syms}ab) \cite{macarthur2007automorphism}. 

\begin{figure}
    \centering
    \includegraphics[width=0.46\textwidth]{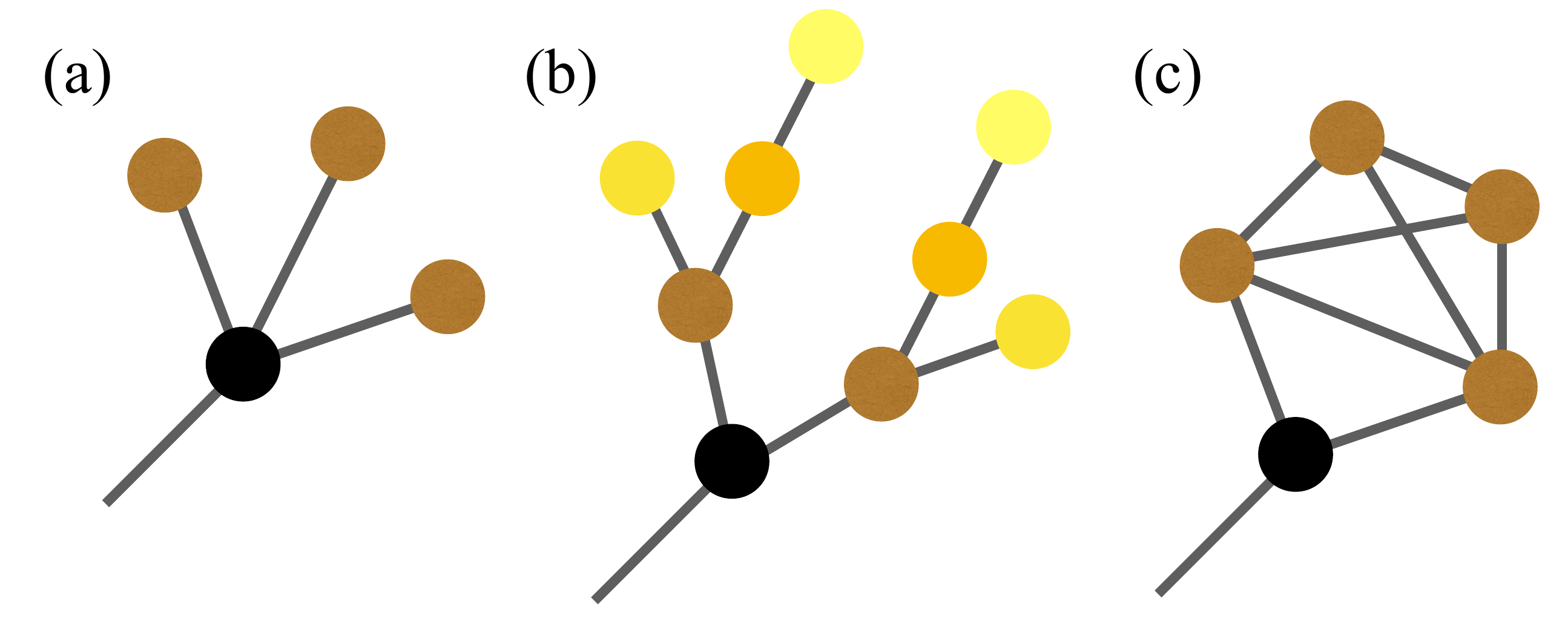}
    \caption{\scriptsize
    Examples of local symmetries in networks. 
    }
    \label{fig:local_syms}
\end{figure}

\begin{figure}
    \centering
    \includegraphics[width=0.48\textwidth]{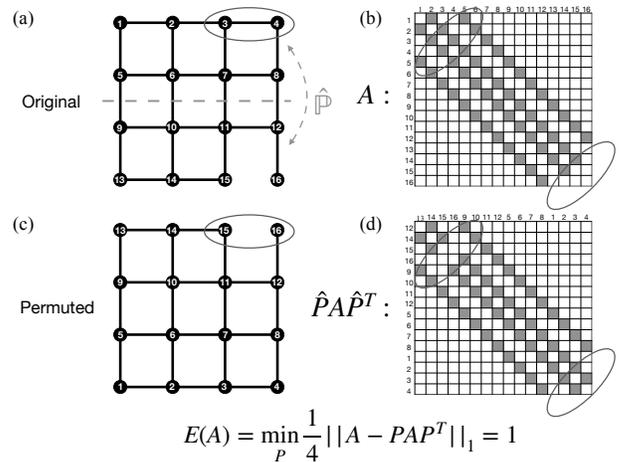}
    \caption{\scriptsize 
    (a): A 2D lattice network with one missing edge, and its node permutation $\hat{\mathbbm{P}}$ that switches the upper and lower half of the network. 
    (b): The adjacency matrix $A$ of the network in (a).
    (c): The permuted network. The circled edge in (a) is not preserved in this node permutation. 
    (d): The adjacency matrix of the permuted network in (c), $\hat{P}A\hat{P}^T$, where $\hat{P}$ is the permutation matrix that encodes the node permutation $\hat{\mathbbm{P}}$. 
    There are for different elements between $A$ and $\hat{P}A\hat{P}^T$ (circled), which gives rise to $E(A) = 1$.
    }
    \label{fig:EA_exmp}
\end{figure}

% Most literature on measuring the symmetry of networks uses graph automorphism algorithms to find the automorphism group, and defines either the number of elements in the automorphism group or the number of nodes involved in the automorphism group as a measure of the level of symmetry of networks. 
However, this type of measure has two problems. First, it is a measure of only exact symmetry, meaning that it fails to capture a large group of approximate symmetries of networks (for example, a lattice with one missing edge still has a high level of symmetry, intuitively speaking, but according to the present symmetry measure it may have zero symmetry, see Figure~\ref{fig:4x4_2D_Lattices}).
For large random graph and real-world networks, an exact symmetry of a large number of nodes imposes strict limitations on network structure that are hard to be satisfied due to the stochasticity of networks.
second, in real-world networks, most of the automorphism transformations found are composed of local permutations \cite{macarthur2007automorphism,xiao2008emergence} (Figure~\ref{fig:local_syms}), for example two degree-one nodes attached to the same node are symmetric to each other. 
Therefore if a network has many automorphisms, it could mean that the network has a lot of local symmetries, instead of global symmetries. 
In contrast, a global symmetry involves a permutation of most nodes in the network.
For example, as Figure~\ref{fig:small_2D_Lattice_examples} shows, a ring lattice as rotational symmetry, a 2D square lattice has mirror and reflectional symmetries, and the network in Figure~\ref{fig:small_2D_Lattice_examples}(c) has mirror symmetry, all of which are global symmetries (Figure~\ref{fig:small_2D_Lattice_examples}). 
Global symmetries of networks are harder to be captured by measures that utilize network automorphism which detect only perfect symmetries, due to the strict requirements that a perfect symmetry of a large number of nodes and edges imposes on the network structure.

To overcome these problems, we introduce a new measure of the level of symmetry of a given network with adjacency matrix $A$:
\begin{equation}
    E(A) = \frac{1}{4} \min_P (|| A - PAP^T ||_1),
\end{equation}
where %$Q$ is a permutation operator that has one non-zero element in each row and column, with properties $\det(Q) = 1$ and $QQ^T = I$, and 
$||\cdot ||$ is the order-1 norm, and $P$ is a node permutation matrix with the property of $\tr{P} = 0$. 
The higher level of symmetry a network has, the lower its $E$ is. For a network with perfect symmetry, it has $E=0$. $P$ acting on the adjacency matrix $A$ on the left-hand-side $PA$ permutes the rows of $A$, and $AP^T$ permutes the columns of $A$. Together $PAP^T$ is the resulting adjacency matrix after a permutation of the node sequence. 
A network's symmetry can be represented by $\hat{P} = \underset{P}{\argmin} || A - PAP^T ||_1$, which encodes the node permutation $\hat{\mathbbm{P}}$ that preserves the adjacency matrix to the highest extent. 
To restrict the symmetries we find to global symmetries, we limit $P$ such that $\tr{P} = 0$, meaning that all nodes are involved in the permutation. 
For unweighted networks, $E$ measures the number of edges that are not preserved after the node permutation ($A_{ij}=1$, $A_{\hat{\mathbbm{P}}_i \hat{\mathbbm{P}}_j} = 0$ for $i,j\in \{ 1,2..N\}$). 
For example, Figure~\ref{fig:EA_exmp}(a) shows a small 2D lattice network with one missing edge. 
Its optimal node permutation $\hat{\mathbbm{P}}$ switches the upper and lower half of the network, turning the network into Figure~\ref{fig:EA_exmp}(c).
In this node permutation, one edge in the original network in Figure~\ref{fig:EA_exmp}(a) (circled) is not preserved (Figure~\ref{fig:EA_exmp}(c), circled), which is reflected by the difference between their adjacency matrices $A$ and $\hat{P}A\hat{P}^T$ (Figure~\ref{fig:EA_exmp}(b,d), circled).
Therefore, this network has $E(A)=1$.
% The permuted network is shown in Figure~\ref{fig:small_net_examples}(b). The edges that are preserved in the permutation are colored in blue, and the two edges that contribute to $E$ are colored in black. Figure~\ref{fig:small_net_examples}(c) shows a ring lattice network with $10$ nodes and $20$ edges. Figure~\ref{fig:small_net_examples}(d) shows a permutation of this ring lattice that is un-optimal -- none of the edges are preserved after the permutation ($E=20$). 
In the figure, the node labels represent the nodes identities, and the location of nodes represent the orders of nodes (the order of the rows and columns in the adjacency matrix that represent nodes). 
Since all nodes are permuted, there is no node that remains at the same location after permutation. 
The edges that are preserved have corresponding edges at the same locations in the permuted network (note though that the corresponding edges at the same locations are not the same edges), while the ones that are not do not. 
A network with perfect global symmetry ($E=0$) has at least one global node permutation that preserves all edges, for example the ring lattice network in Figure~\ref{fig:small_2D_Lattice_examples}(a) has perfect symmetry, with permutations $\hat{\mathbbm{P}}$ that preserves all edges ($E=0$), one of which shown in Figure~\ref{fig:small_2D_Lattice_examples}(b). 
We also define $\mathcal{E}(A, P) = \frac{1}{4} || A - PAP^T ||_1$. Therefore we have $E(A) = \min_P \mathcal{E}(A, P) = \mathcal{E}(A, \hat{P})$.

% \begin{figure*}
%     \centering
%     \includegraphics[width=0.3\textwidth]{figs/2D_Lattice_N100_before_perm.png}
%     \includegraphics[width=0.3\textwidth]{figs/2D_Lattice_N100_after_perm_E0.png}
%     \caption{Caption}
%     \label{fig:small_2D_Lattice_examples}
% \end{figure*}

In order to compare the level of symmetry of networks with different sizes, we normalize $E$ by the maximum possible $\mathcal{E}(A, P)$ for a network with $N$ nodes ($ \mathcal{E}_{\max}=\frac{1}{2} {N\choose 2}$ when the network has $\frac{1}{2} {N\choose 2}$ edges and $P$ such that no edges are preserved):
\begin{equation}
    \mathcal{S}(A) = \frac{E(A)}{\frac{1}{2} {N\choose 2}} = \frac{\min_P (|| A - PAP^T ||_1)}{N(N-1)} .
\end{equation}
% In the rest of the paper we use $S$ to measure the level of symmetry of networks.
By definition, empty networks (number of edges $M=0$) and fully-connected networks ($M={N\choose 2}$) always have perfect symmetry ($\mathcal{S}=0$). 

\section{Network Symmetry for several network models}

In this section we investigate and compare the symmetry of networks generated from several classic network models, providing insights on network ensemble structures.

% Before we go into details about $S$ for different networks, we establish a baseline for the value of $S$: assuming the nodes in a network are uncorrelated, and given large-enough network size, the expected value for $\frac{1}{4}|| A - QAQ^T ||_1$ is $M(1-\frac{M}{{N\choose 2}})$. Therefore, we define $S_{rand} = \frac{4M}{N(N-1)} \left[ 1 - \frac{2M}{N(N-1)} \right]$ as a baseline value for $S$. 

\subsection{Erdös-Rényi Networks}

\begin{figure}
    \centering
    \includegraphics[width=0.48\textwidth]{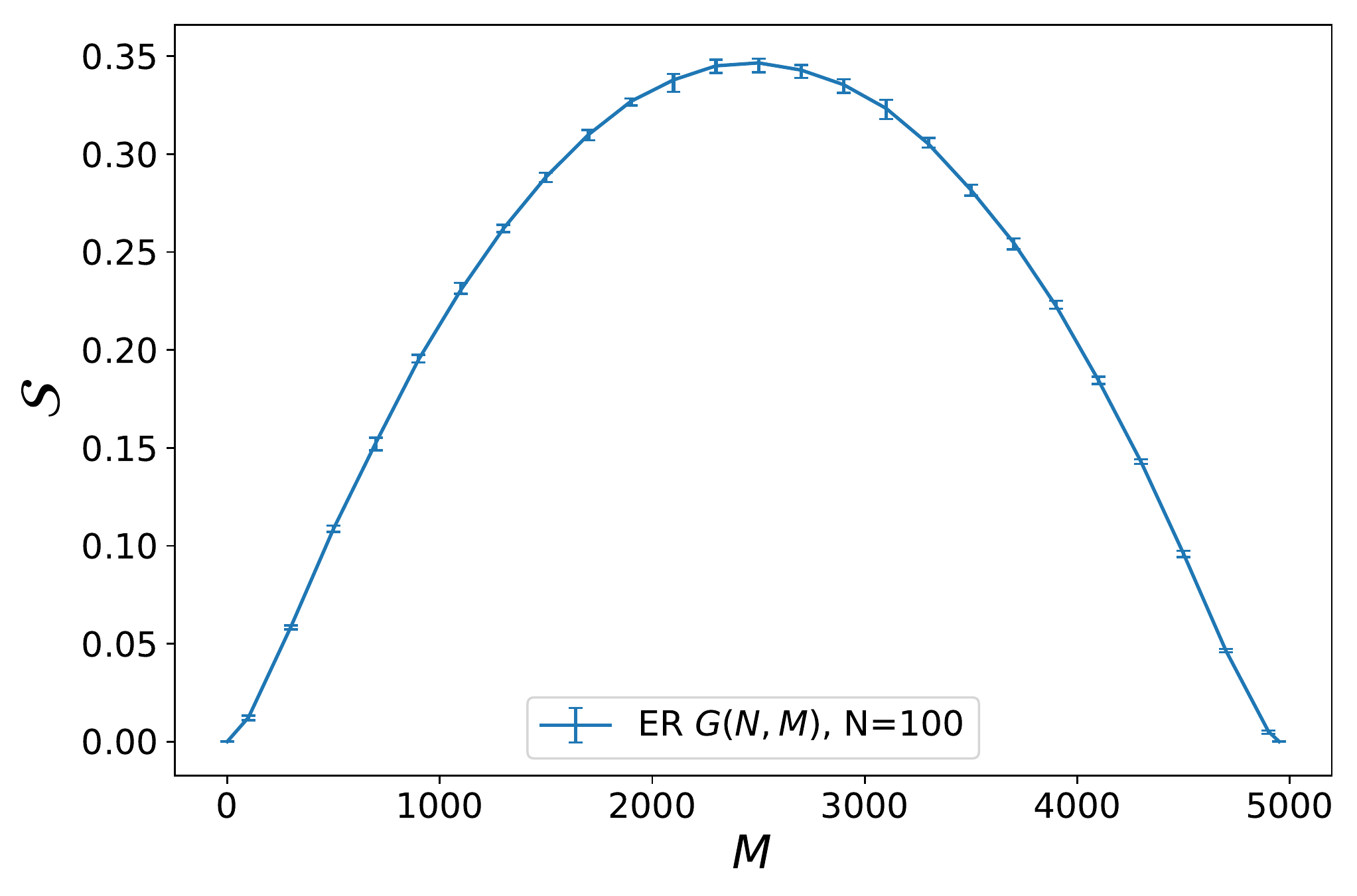}
    \caption{\scriptsize
    $\mathcal{S}$ of ER networks sampled from $G(N, M)$ ensembles with $N=100$ and $M\in [0, \frac{N(N-1)}{2}]$. 
    The error bars mark the range of  $\mathcal{S}$ over which different samples from ensembles span. 
    % The orange curve represents $\mathcal{S}_{rand}$ for the same $N$ and $M$. 
    }
    \label{fig:S_ER_N100}  
\end{figure}

Figure~\ref{fig:S_ER_N100} shows $\mathcal{S}$ measured for Erdös-Rényi (ER) networks \cite{erdos1960evolution} sampled from $G(N,M)$ ensembles with $N=100$ and different $M\in [0, 4,950]$ values.
% , as well as $S_{rand}$. 
When the ER networks are empty ($M$ = 0) or complete ($M$ = 4,950), they have perfect symmetry, which is reflected by $\mathcal{S} = 0$. When half of the potential edge slots are filled ($M = 2,475$), the networks have the least level of symmetry (highest $\mathcal{S}$), due to the high uncertainty of edge placements. Interestingly, the $\mathcal{S}$ versus $M$ curve for ER networks with a fixed $N$ is symmetric on the two sides, suggesting that the expected level of symmetry for an ER network with $N$ nodes and $M$ edges is the same as that of an ER network with $N$ nodes and ${N\choose 2} - M$ edges. 
This can be explained using the following theorem:

\begin{theorem}\thlabel{theo:complement}
    For any given unweighted network with adjacency matrix $A$, its complementary network (each edge turns into a non-edge, and each non-edge turns into an edge) with adjacency matrix $A'$ has the same level of symmetry as its complement ($E(A) = E(A')$, and $\mathcal{S}(A) = \mathcal{S}(A')$). 
\end{theorem}
\begin{proof}
    For any given unweighted network with adjacency matrix $A$, its network complement's adjacency matrix can be expressed by $A' = \mathbbm{1} - A - I$, where $\mathbbm{1}$ represents an $N\times N$ matrix with all elements equal to one, and $I$ is the identity matrix. 
Therefore, for any permutation matrix $P$, we have 
\begin{align}
    &\mathcal{E}(A', P) = \frac{1}{4} || A' - PA'P^T ||_1  \nonumber \\
    &= \frac{1}{4} || \mathbbm{1} - A - I - P(\mathbbm{1} - A - I)P^T ||_1.
    \label{eq:mathcal_e_ap}
\end{align}
Utilizing the properties of $P$, we can simplify eq.~\ref{eq:mathcal_e_ap}: because the each row and column of $P$ has one non-zero element equal to one, $P\mathbbm{1} = \mathbbm{1}$, therefore $P\mathbbm{1}P^T = \mathbbm{1}$; because of $PP^T=I$, we have $PIP^T = PP^T = I$. Therefore we have
\begin{align}
    \mathcal{E}(A', P) &= \frac{1}{4} || \mathbbm{1} - A - I - \mathbbm{1} + PAP^T + I) ||_1 \nonumber \\
    &= \frac{1}{4} || -A + PAP^T ||_1 \nonumber \\
    &= \frac{1}{4} || A - PAP^T ||_1 \nonumber \\
    &= \mathcal{E}(A, P).
\end{align}
Therefore,
\begin{align}
    E(A') &=  \min_P \mathcal{E}(A', P) = \min_P \mathcal{E}(A, P) \nonumber  \\
    &= E(A),
\end{align}
and
\begin{equation}
    \mathcal{S}(A') = \frac{E(A')}{\frac{1}{2} {N\choose 2}} = \frac{E(A)}{\frac{1}{2} {N\choose 2}} = \mathcal{S}(A).
\end{equation}
\end{proof}

% This result is in consistence with the fact that two ER networks sampled from $G(N, M)$ and $G(N, {N\choose 2} - M)$ ensembles have the same likelihood, due to the fact that the two ensembles have the same number of elements.  
For any ER network in ensemble $G(N, M)$, its complementary network can be found in ensemble $G(N, {N\choose 2} - M)$, and vise versa. 
Moreover, there are ${{N\choose 2} \choose M}$ elements in ensemble $G(N, M)$, and there are ${{N\choose 2} \choose {N\choose 2} - M} = {{N\choose 2} \choose M}$ elements in ensemble $G(N, {N\choose 2} - M)$, which means that the two ensembles have the same number of elements. 
Therefore, the networks in ensemble $G(N, M)$ and the networks in ensemble $G(N, {N\choose 2} - M)$ have a one-to-one mapping relation, with each pair of networks being complementary to each other. 
According to Theorem 1, each pair of complementary networks in ensembles $G(N, M)$ and $G(N, {N\choose 2} - M)$ have the same $\mathcal{S}$, therefore the ensemble average of $\mathcal{S}$ for networks in ensembles $G(N, M)$ and $G(N, {N\choose 2} - M)$ are the same.

% $S$ for ER networks are less than $S_{rand}$, mostly around $0.7$ of the latter. 

\subsection{Random Geometric Graphs}

% \begin{figure*}
%     \centering
%     \includegraphics[width=0.45\textwidth]{figs/RGG_P_S_vs_M_N100.png}
%     \includegraphics[width=0.45\textwidth]{figs/RGG_S_vs_r_P.png}
%     \caption{Caption}
%     \label{fig:RGG_P_S_vs_M_N100}
% \end{figure*}

The nodes in random geometric graphs (RGGs) \cite{gilbert1961random} have intrinsic embeddings, which govern the network structure, and provide a certain level of symmetry to the network. Figure~\ref{fig:RGG_P_S_vs_M_N100}(a) shows $\mathcal{S}$ for RGGs (both with and without periodic boundary conditions) with $N=100$ and different connection radii, giving rise to a wide range of $M$. $\mathcal{S}$ for RGGs is less than that of ER networks with the same $N$ and $M$, suggesting that RGGs are more symetric compared to ER networks. This is to be expected due to the latent geometric embedding of nodes in RGGs. Nodes only connect to other nodes that are located in its vicinity, rendering a high similarity between nodes, resulting in a lattice-like space-filling structure with a high level of symmetry. In the range of $M\in [0,800]$, the symmetry levels of RGGs with or without periodic boundary conditions are similar, because when the connection radius is short, the effect of the periodic boundary conditions is negligible. In the range of $M\in [800, 4950]$, $\mathcal{S}$ for RGGs with periodic boundary conditions is noticeably higher than RGGs without periodic boundary conditions. 
Figure~\ref{fig:RGG_P_S_vs_M_N100}(b) shows $\mathcal{S}$ for RGGs (with and without periodic boundary conditions) with $N=100$ and various connection radii $r$. As $r$ increases, $\mathcal{S}$ first increases, and then decreases as RGGs approach complete graphs. The highest $\mathcal{S}$ of RGGs with periodic boundary conditions is higher than that of RGGs without periodic boundary conditions. $\mathcal{S}$ of RGGs with periodic boundary conditions converges faster to zero than RGGs without periodic boundary conditions, due to the fact that periodic boundary conditions give rise to a higher number of edges for a given $r$. 

\begin{figure*}
    \centering
    \includegraphics[width=\textwidth]{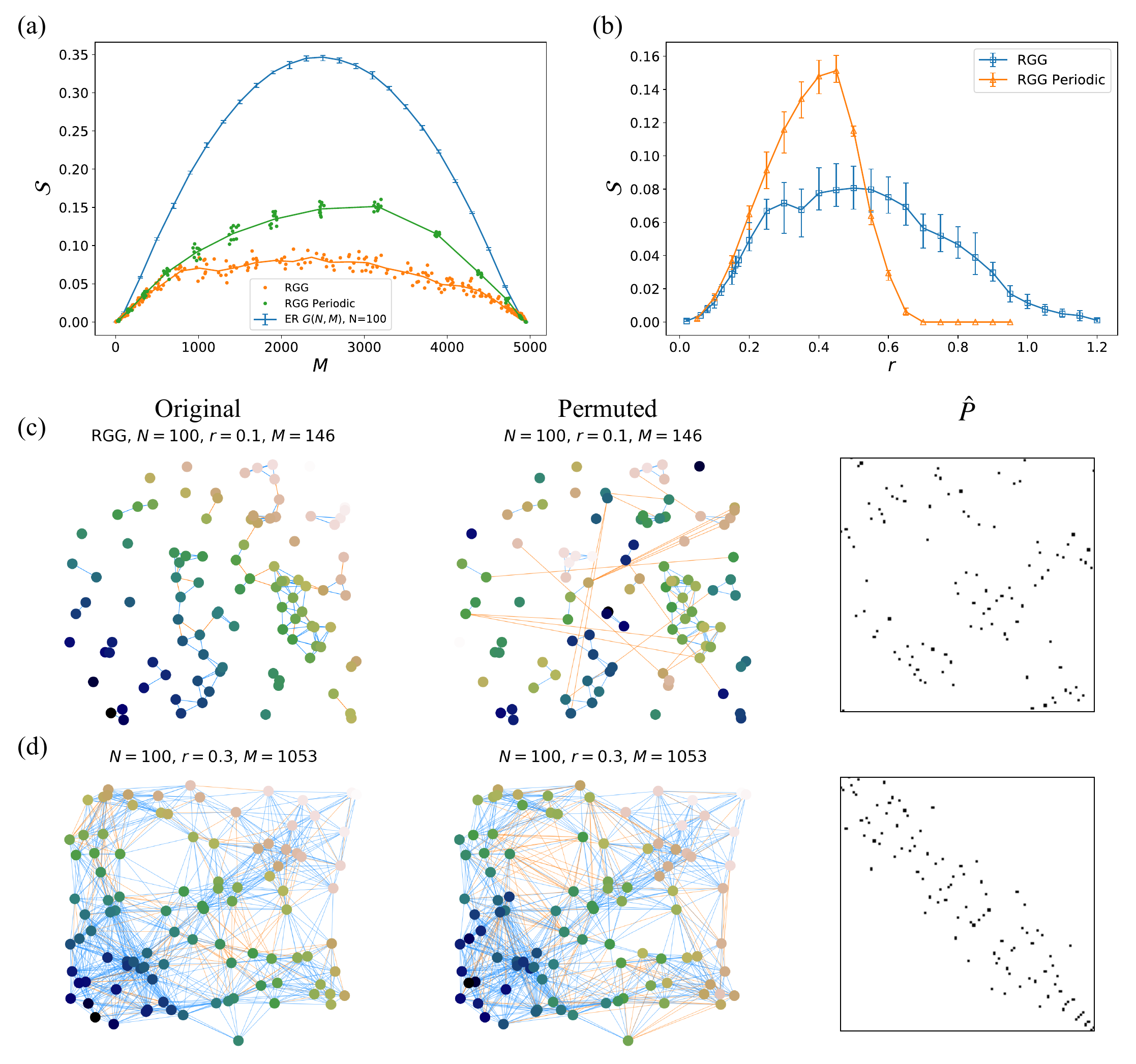}
    \caption{\scriptsize
    Symmetry of RGGs.
    (a): $\mathcal{S}$ for RGGs (with and without periodic boundary conditions) with $N=100$ and different number of edges $M$, compared to that of ER networks with the same sizes. In general, for same-sized and same-number-of-edges networks, RGGs without periodic boundary conditions are more symmetric than RGGs with periodic boundary conditions, and in both cases RGGs are more symmetric than ER networks with the same $N, M$.
    (b): $\mathcal{S}$ for RGGs (with and without periodic boundary conditions) with $N=100$ and various connection radii $r$. As $r$ increases, $\mathcal{S}$ first increases, and then decreases as RGGs approach complete graphs. 
    (c,d): Two examples of RGGs (no periodic boundary conditions) with $N=100$, $r=0.1$ and $r=0.3$ respectively, and their symmetry (optimal permutation matrices $\hat{P}$). The original networks show the networks with nodes arranged by their intrinsic 2D embedding in unit squares, colored in a gradient according to their locations. In the permuted networks the nodes are permuted according to the optimal permutation matrices $\hat{P}$, with the same colors as in the original networks, and the edges are colored in blue if they are preserved in the node permutation, and orange if not. 
    The rows and columns in $\hat{P}$ are ordered according to the corresponding nodes' proximity to the lower left corner of the unit square. 
    % Therefore, a cluster of non-zero elements in $\hat{P}$ indicates a permutation of a subset of nodes that are closely embedded in space, and a non-zero element close to the diagonal indicates that a node is permuted to another node that is embedded closely in space. 
    }
    \label{fig:RGG_N100_before_after_perm}
    \label{fig:RGG_P_S_vs_M_N100}
\end{figure*}

Figure~\ref{fig:RGG_N100_before_after_perm}(c,d) show two examples of RGGs with $N=100$ and connection radius $r=0.1$, $0.3$ respectively, and the permuted RGGs according to $\hat{P}$. The blue edges are the edges that are preserved in the permutations, and the orange ones are not. 
The rows and columns in $\hat{P}$ are ordered according to the corresponding nodes' proximity to the lower left corner of the unit square. 
Therefore, 
% a cluster of non-zero elements in $\hat{P}$ indicates a permutation of a subset of nodes that are closely embedded in space, and a non-zero element close to the diagonal indicates that a node is permuted to another node that is embedded closely in space. 
the several clusters of non-zero elements in $\hat{P}$ in Figure~\ref{fig:RGG_N100_before_after_perm}(c) indicate that there are subsets of nodes that are located in proximity being permuted to another location as a whole, and 
the non-zero elements in $\hat{P}$ in Figure~\ref{fig:RGG_N100_before_after_perm}(d) concentrated along the diagonal indicate that all nodes are permuted to a location close to its intrinsic locations.

\subsection{Watts-Strogatz Networks}

% \begin{figure}
%     \centering
%     \includegraphics[width=0.20\textwidth]{figs/WS_N50_k10_before_perm.png}
%     \includegraphics[width=0.20\textwidth]{figs/WS_N50_k10_after_perm.png}
%     \caption{Caption}
%     \label{fig:WS_N50_k10_before_after_perm}
% \end{figure}

% \begin{figure*}
%     \centering
%     \includegraphics[width=0.45\textwidth]{figs/WS_S_vs_M_N100.png}
%     \includegraphics[width=0.45\textwidth]{figs/WS_S_vs_p_N100.png}
%     \caption{Caption}
%     \label{fig:WS_S_M_N100}
% \end{figure*}

\begin{figure*}
    \centering
    \includegraphics[width=\textwidth]{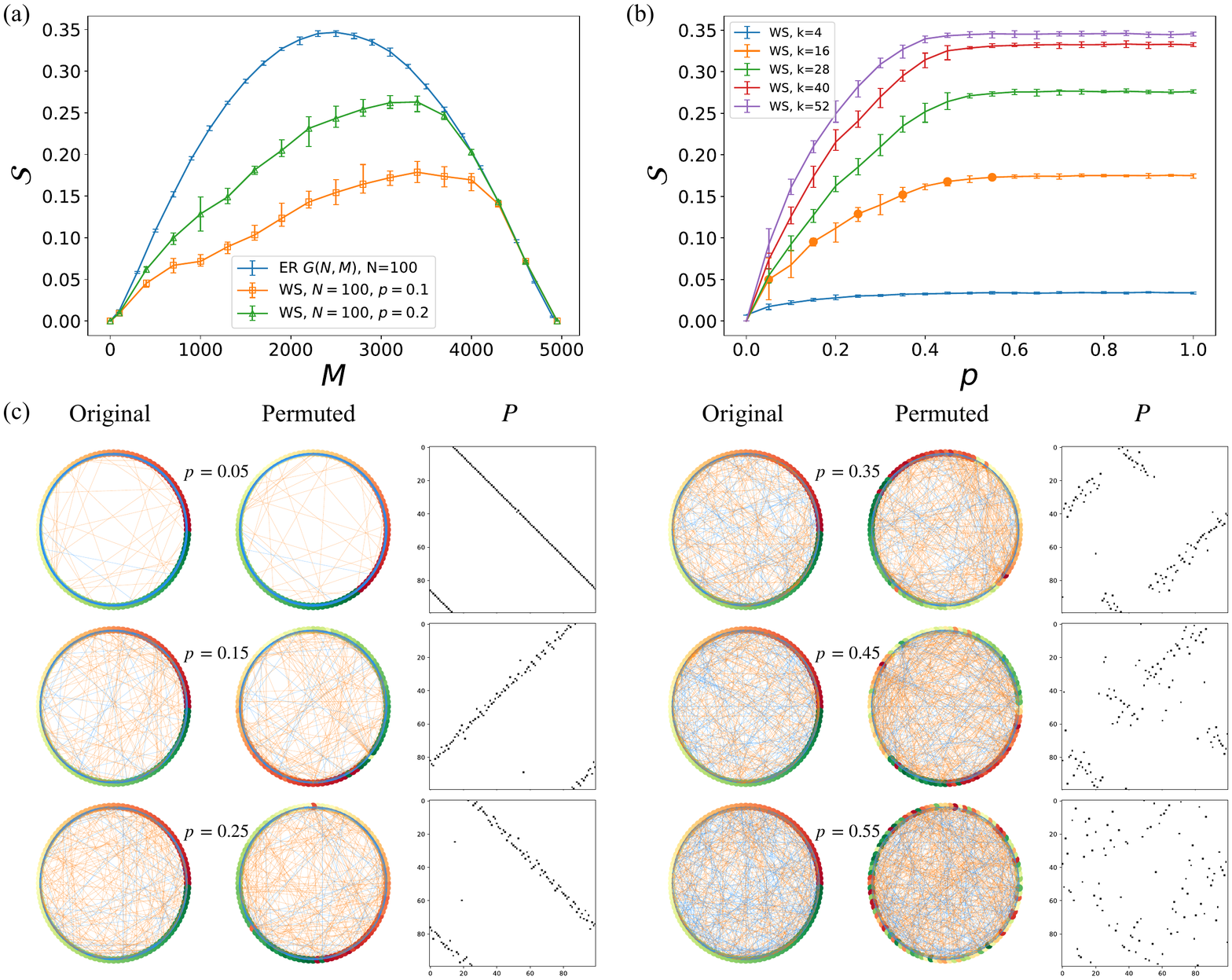}
    \caption{\scriptsize
    Symmetry of WS networks.
    (a): $\mathcal{S}$ for WS networks with $N=100$ and rewiring probabilities $p=0.1$ and $p=0.2$ and different number of edges $M$, compared to that of ER networks with the same sizes. In general, for same-sized and same-number-of-edges networks, WS networks with $p=0.1$ are more symmetric than WS networks with $p=0.2$. As $M$ increases, the level of symmetry of WS networks converges to that of ER networks with the same $N$ and $M$.
    (b): $\mathcal{S}$ for WS networks with $N=100$ and five different $M$, and various rewiring probabilities $p$. As $p$ increases, WS networks with fixed $N,M$ become more asymmetric, until they reaches the same level of symmetry as that of ER networks with the same $N,M$.
    (c): Six examples of WS networks and their symmetry (optimal permutation matrices $\hat{P}$), sampled from the dotted data points in (b). The original networks show the networks with nodes arranged by their intrinsic order on a ring, colored according to a gradient on node orders. In the permuted networks the nodes are permuted according to the optimal permutation matrices $\hat{P}$, with the same colors as in the original networks, and the edges are colored in blue if they are preserved in the node permutation, and orange if not.}
    \label{fig:WS}
    \label{fig:WS_S_M_N100}
\end{figure*}

% Watts-Strogatz model (WS) was introduced .... .
In the Watts-Strogatz model \cite{watts1998collective}, when the rewiring probability $p=0$, the networks are ring lattices, which have perfect symmetry (Figure~\ref{fig:WS_N50_k10_before_after_perm}(a)).
As the rewiring probability $p$ increases, the orderliness of the network structure is gradually compromised by randomness, and accordingly the symmetry of the network decreases (Figure~\ref{fig:WS_S_M_N100}(b)), therefore $\mathcal{S}$ increases, until it becomes indistinguishable from that of a typical ER network with the same $N$ and $M$. 
For WS networks with $N=100$, their levels of symmetry become the same as ER networks with the same $N$ and $M$ at around $p=0.5$. 
Figure~\ref{fig:WS_S_M_N100}(a) shows $\mathcal{S}$ for WS networks with $N=100$, $p=0.1, 0.2$ and different $M$, compared with $\mathcal{S}$ for ER networks. 
We can see that first of all, WS networks with $p=0.1$ have smaller $\mathcal{S}$ than WS networks with $p=0.2$, which is expected because of the increase in randomness brought by the higher the rewiring probability. 
Second, when $M$ is small, $\mathcal{S}$ of WS networks is less than that for ER networks. As $M$ increases, $\mathcal{S}$ of WS networks gets closer to that for ER networks, until it becomes the same. 
For WS networks with $p=0.1$, the threshold of $M$ beyond which $\mathcal{S}$ is the same as that of ER networks (roughly $M=4,300$) is greater compared to that of WS networks with $p=0.2$ (roughly $M=4,000$), which 
can be explained by the fact that when $p$ is smaller, the number of rewired edges is less for a given $M$, therefore it takes a larger value of $M$ to reach the randomness (symmetry level) of equivalent (same $M$) ER networks.

Figure~\ref{fig:WS}(c) shows six examples of WS networks sampled from dotted data points in Figure~\ref{fig:WS}(b), with $N=100$, $\left< k\right> =16$ and various $p$ values. 
The original networks are shown, with nodes colored in the order of the ring lattice network before edge rewiring. The symmetry in these WS networks are found, encoded in node permutations $\hat{\mathbbm{P}}$ with permutation matrices $\hat{P}$.
The networks after permutations are shown, with preserved edges (blue) at the same locations as in the original networks, and the unpreserved edges (orange) at different locations as in the original networks. 
At $p=0.05$, the nodes in the permuted network are rotated almost perfectly except for two nodes, reflected by its $\hat{P}$. Most of the rewired edges are not preserved, and most of the unrewired edges are preserved (on the brim of the circle). 
The number of slots by which the nodes are rotated is determined by the rewired edges -- a rotation by any number of slots will preserve most of the unrewired edges, therefore preserving the rewired edges to the maximum extent is the determinant factor of by how many slots the nodes are rotated as a whole, for low $p$. 
At $p=0.15$, more edges are rewired, hence more edges are unpreserved in the permutation. 
The optimal permutation found includes an approximate reflection and rotation, accompanied by local disruptions and flips, reflected in its $\hat{P}$. 
At $p=0.35$ and $p=0.45$, we observe that the optimal permutations are no longer approximately a global rotation or reflection of nodes. Instead, the nodes are separated into several regions and each region shows a different permutation pattern, which are the combinations of shifts and rotations and reflections and local disruptions. 
At $p=0.55$, the WS networks' $\mathcal{S}$ has converged to that of equivalent ER networks, and therefore in its optimal permutation the original order of the nodes is completely broken.

\subsection{ Stochastic Block Models}
\label{sec:SBM}

% \begin{figure*}
%     \centering
%     \includegraphics[width=0.4\textwidth]{figs/SBM_S_vs_pout:pin_diff_M_with_thresh.png}
%     \includegraphics[width=0.4\textwidth]{figs/SBM_S_vs_M.png}
%     \caption{Caption}
%     \label{fig:SBM_S_vs_pout:pin_diff_M}
% \end{figure*}

\begin{figure*}
    \centering
    \includegraphics[width=\textwidth]{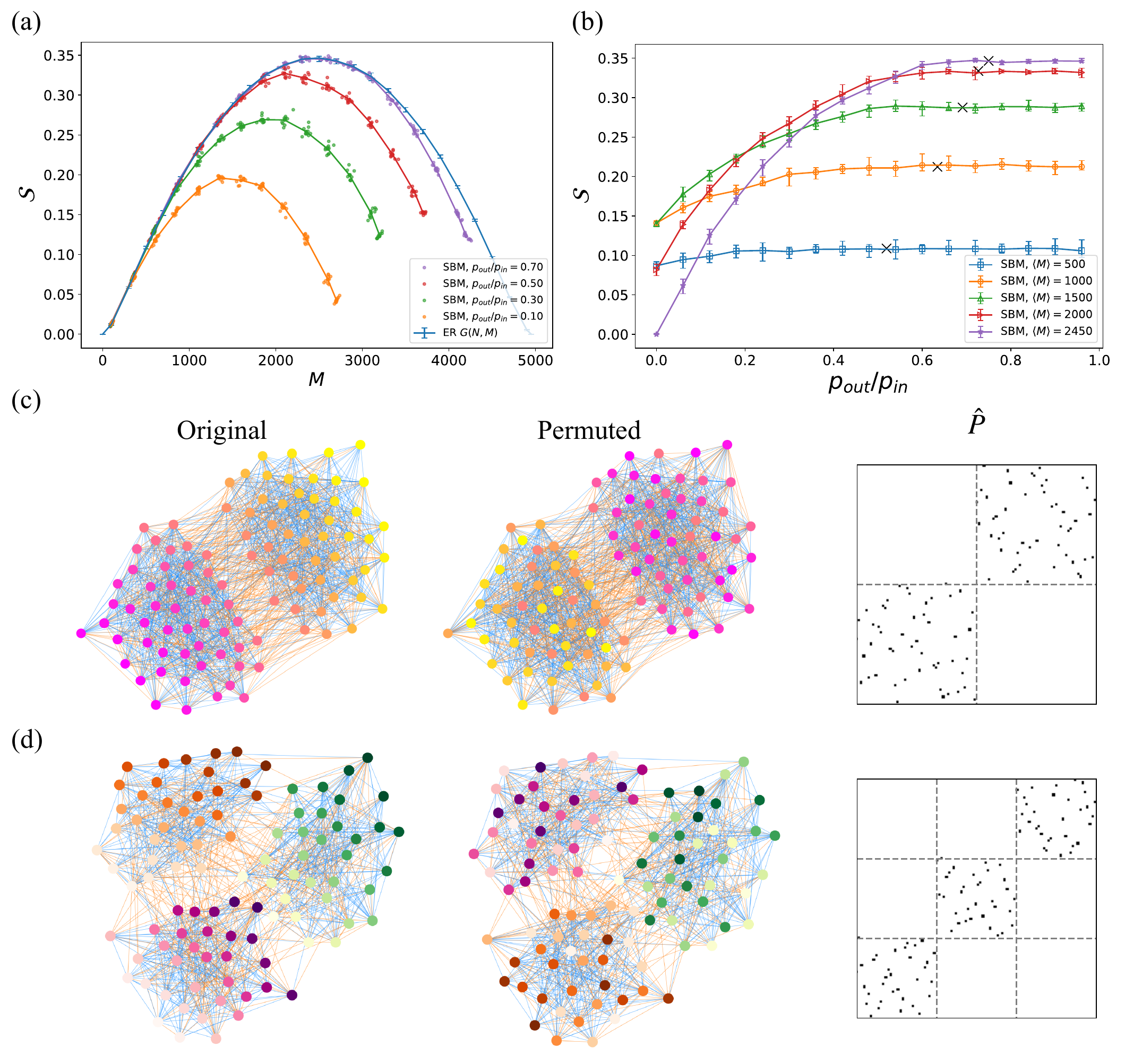}
    \caption{\scriptsize Symmetry of SBM networks.
    (a): $\mathcal{S}$ for SBM networks with $N=100$, several fixed $\frac{p_{out}}{p_{in}}$, and $p_{in}$ ranging from $0$ to $1$, which gives rise to a wide range of $M$. The maximum expected $M$ for SBM networks is limited by $\frac{p_{out}}{p_{in}}$ due to the fact that $p_{in}\leq 1$. In general, for same-sized and same-number-of-edges networks, SBM networks are more symmetric than ER networks.
    (b): $\mathcal{S}$ for SBM networks with $N=100$, several fixed expected number of edges and various $\frac{p_{out}}{p_{in}}$. As $\frac{p_{out}}{p_{in}}$ increases, $\mathcal{S}$ increases until it reaches that of ER networks with the same $N, M$.
    The black ``$\times$'' mark the detectability transition threshold of SBM networks, beyond which the SBM networks are indistinguishable from ER networks.
    (c,d): Two examples of SBM networks with two and three equal-sized communities respectively, and their symmetry (optimal permutation matrices $\hat{P}$). The original networks show the networks with nodes colored according to their assigned community labels. In the permuted networks the nodes are permuted according to the optimal permutation matrices $\hat{P}$, with the same colors as in the original networks, and the edges are colored in blue if they are preserved in the node permutation, and orange if not. 
    The rows and columns in $\hat{P}$ are ordered according to the nodes' assigned community labels. 
    % Therefore, a cluster of non-zero elements in $\hat{P}$ indicates a permutation of a subset of nodes that are closely embedded in space, and a non-zero element close to the diagonal indicates that a node is permuted to another node that is embedded closely in space. 
    }
    \label{fig:SBM_Q2_Q3_examples}
    \label{fig:SBM_S_vs_pout:pin_diff_M}
\end{figure*}

The Stochastic Block Model \cite{holland1983stochastic} generate networks with modular structures, with $q$ pre-defined modules, and $p_{a,b}$ ($a,b \in \{ 1,2..,q \}$) as the probability of a node in module $a$ and a node in module $b$ being connected.
Here for simplicity, we restrict the model such that within-module connection probability $p_{in}$ is the same for all modules, and between-module connection probability $p_{out}$ is the same for all pairs of modules. 
We want to understand the effect of the existence of modules on network symmetry. 
Imagine if a network is composed of two connected modules that are identical to each other, then intuitively this network would have a high level of symmetry. 
However, most commonly the modules in networks are not identical. Nevertheless, modularity can still bring a certain level of symmetry to the network structure. 
Figure~\ref{fig:SBM_S_vs_pout:pin_diff_M}(b) shows $\mathcal{S}$ as a function of $\frac{p_{out}}{p_{in}}$ for SBM networks with two same-sized modules, $N=100$ and different $\left< M\right>$ ($\left< M\right>=2 {\frac{N}{2}\choose 2}p_{in} + \left( \frac{N}{2} \right)^2 p_{out}$). 
A smaller $\frac{p_{out}}{p_{in}}$ value indicates a stronger modular structure. 
We can see that as $\frac{p_{out}}{p_{in}}$ increases, $\mathcal{S}$ increases, suggesting that the symmetry level decreases, until $\mathcal{S}$ reaches the same value as that of a typical ER network with the same $N$ and $M$. 
At $\frac{p_{out}}{p_{in}} = 0$, each SBM network is composed of two disconnected equal-sized ER networks. For $\left< M\right>=2,450$ SBM networks at $\frac{p_{out}}{p_{in}} = 0$, $p_{in}=1$, which means the two ER networks are both fully connected, therefore they have $\mathcal{S}=0$. 
Figure~\ref{fig:SBM_S_vs_pout:pin_diff_M}(a) shows $\mathcal{S}$ as a function of $M$ for four different fixed $\frac{p_{out}}{p_{in}}$ values, compared against that of ER networks. The maximum expected $M$ for SBM networks with different $\frac{p_{out}}{p_{in}}$ are limited due to $p_{in}\leq 1$. We can see that for a fixed $\frac{p_{out}}{p_{in}}$, $\mathcal{S}$ first increases as $M$ increases, indistinguishable from $\mathcal{S}$ of ER networks, then it deviates from $\mathcal{\mathcal{S}}$ of ER networks, reaches its peak, and then decreases. Lower $\frac{p_{out}}{p_{in}}$ indicates clearer community structure, which in turn results in a higher level of symmetry (lower $\mathcal{S}$). 

SBM networks have detectability thresholds -- SBM networks can be indistinguishable from ER networks sampled from ensembles with the same average edge density, and the planted modules in these networks are undetectable \cite{dyer1989solution,condon2001algorithms,krivelevich2006semirandom,krzakala2009hiding,bickel2009nonparametric,decelle2011asymptotic}. For SBM networks with $N$ nodes, $q$ same-sized modules and $p_{a,b}$ ($a,b \in \{ 1,2..,q \}$) equal to $p_{in}$ when $a=b$ and $p_{out}$ when $a\neq b$, the modules are detectable only when $p_{in} - p_{out} > \frac{\sqrt{qp_{in} + q(q-1)p_{out}}}{\sqrt{N}}$ \cite{decelle2011asymptotic}. 
Specifically, for SBM networks shown in Figure~\ref{fig:SBM_S_vs_pout:pin_diff_M}(b) ($q=2$ and several fixed expected number of edges $\left< M \right> = 2 {\frac{N}{2}\choose 2}p_{in} + \left( \frac{N}{2} \right)^2 p_{out}$), their detectability condition translates as $\frac{p_{out}}{p_{in}} < \frac{\sqrt{\frac{2\left< M \right>}{N}}-1}{\sqrt{\frac{2\left< M \right>}{N}}+1}$, marked in the figure by ``$\times$''. 
In Figure~\ref{fig:SBM_S_vs_pout:pin_diff_M}(b) $\mathcal{S}$ saturates to $\mathcal{S}$ of ER networks slightly earlier than the detectability threshold.

Symmetry in SBM networks comes from two sources: symmetry within each community, and symmetry of different communities. For example, for a simple SBM network with equal-sized communities that have the same connection probability $p_{in}$ (Figure~\ref{fig:SBM_Q2_Q3_examples}(c,d)), all communities are equivalent to ER networks with connection probability $p_{in}$, which have a certain level of symmetry determined by $p_{in}$ (Figure~\ref{fig:S_ER_N100}). 
% We can see that most intra-community edges are preserved in the symmetry permutation (blue), and more inter-community edges are not preserved (orange). 
Moreover, the communities are similar to one another, 
% providing a certain level of symmetry to the entire network (Figure~\ref{fig:SBM_Q2_Q3_examples}(c)). 
% [note: add description of permutations that retain node communities, which are also approximate global symmetries.]
% When there are more than two same-sized same-connection-probability communities in an SBM network, the communities are similar to one another, 
therefore the communities interchanging with one another contributes to the symmetry of the network as well (Figure~\ref{fig:SBM_Q2_Q3_examples}(c,d)). 
The symmetry of the SBM networks in Figure~\ref{fig:SBM_Q2_Q3_examples}(c,d) is reflected by their optimal node permutation matrices $\hat{P}$.
The rows and columns in $\hat{P}$ are ordered according to the nodes' assigned community labels. 
Therefore, in Figure~\ref{fig:SBM_Q2_Q3_examples}(c) the non-zero elements in $\hat{P}$ are only located in the off-diagonal blocks, indicating that nodes in the two communities are exchanged in the permutation.
Similarly, in Figure~\ref{fig:SBM_Q2_Q3_examples}(d) the non-zero elements in $\hat{P}$ are located in two off-diagonal blocks and one central block on diagonal, indicating that nodes in community one and three are exchanged in the node permutation and nodes in the second community are scrambled within the community in the node permutation.
% Whether the nodes in the two communities switch places in $\hat{Q}$ is mostly determined by the two competing forces: the symmetry between the two communities, and the symmetry within each community. When the gain from switching communities is greater than the gain from switching nodes within each community, then the communities will not be switched. 

\section{Real-World Networks}

\begin{figure*}
    \centering
    \includegraphics[width=\textwidth]{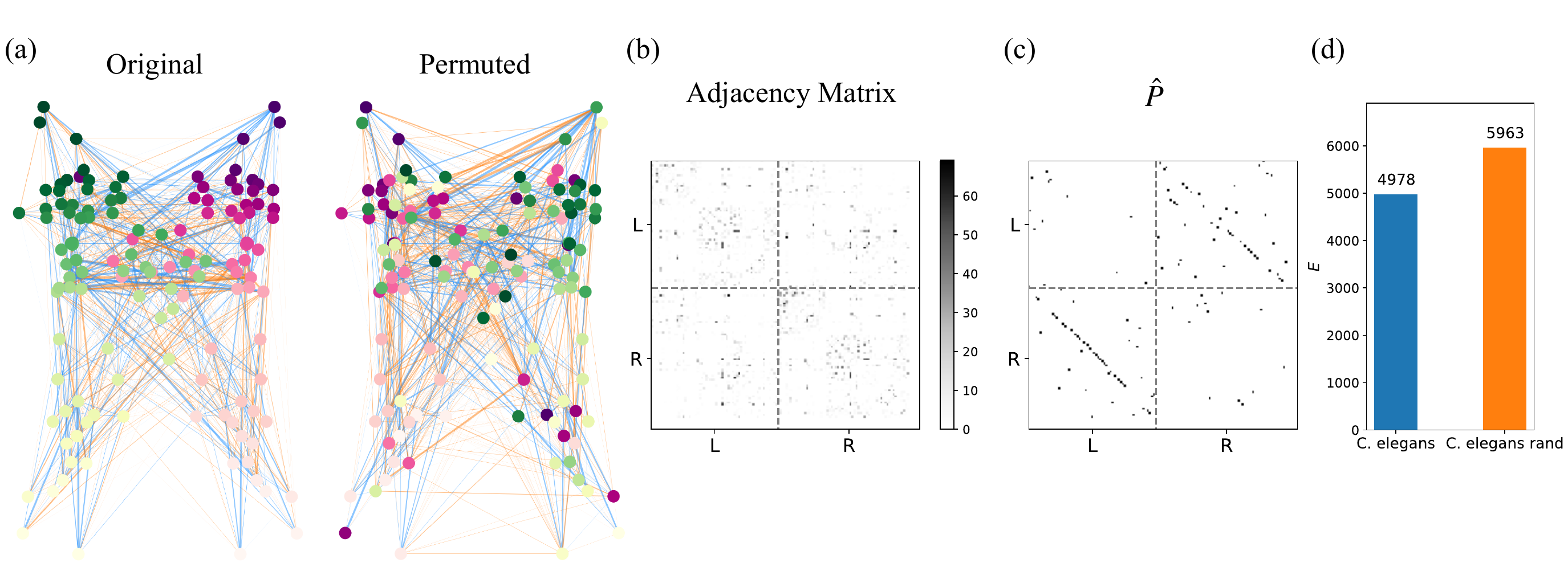}
    \caption{\scriptsize Symmetry of the C. elegans rostral ganglia neural network \cite{cook2019whole, gleeson2010neuroml}. 
    (a): Original neural network \cite{cook2019whole, gleeson2010neuroml} and the network after node permutation. In the original network the locations of the nodes represent the locations of the soma of the neurons in the brain (top view). 
    (b): The adjacnecy matrix of the network \cite{cook2019whole, gleeson2010neuroml}. The network is weighted, and the weights represent the number of synaptic connections between neurons (gap junctions not included). 
    (c): The optimal node permutation matrix $\hat{P}$. In the adjacency matrix in (b) and $\hat{P}$ the rows and columns are ordered according to the hemispheres that the corresponding nodes belong to, and each node and its corresponding node in the other hemisphere (if included in the network) have approximately the same rank among their hemispheres. ``L'' represents the left hemisphere, and ``R'' represents the right hemisphere.
    (d): $E$ of the C. elegans brain network and the network after a degree-preserving randomization. 
    }
    \label{fig:C_elegans}
\end{figure*}

Most brains are composed of two hemispheres, providing natural bilateral symmetry to different types of brain networks. 
Since most brain networks are weighted, we extend $E$ to weighted network, which then represents the least possible sum of the weights of the edges that are not preserved in any global node permutation. 
However, the normalization in the definition of $S$ is no longer applicable to weighted networks, and since there is no good normalization for $E$ of weighted networks, we will use $E$ to measure the level of asymmetry of networks in this section. 

\subsection{C. elegans brain network}

Caenorhabditis elegans is one of the most studied species in network science as well as neural science, due to its completely mapped-out nervous system \cite{white1986structure,hall2008introduction}. 
Here we use our network symmetry algorithm to investigate the symmetry of the C. elegans rostral ganglia neural network with $N=131$ nodes and $M=1,105$ weighted edges (Figure~\ref{fig:C_elegans}(a,b), neurons in the left and right hemispheres are colored differently) \cite{cook2019whole, gleeson2010neuroml}. 
We use the locations of the soma of neurons to represent node locations in the network layout.
We can see that the layout of the C. elegans neural network reflects the bilateral symmetry of the system in physical space. 
$\hat{P}$ of the network is portrayed in Figure~\ref{fig:C_elegans}(c), which results in $E=4,978$. 
In the adjacency matrix in Figure~\ref{fig:C_elegans}(b) and $\hat{P}$ in Figure~\ref{fig:C_elegans}(c) the rows and columns are ordered according to the hemispheres that the corresponding nodes belong to, and each node and its counterpart in the opposite hemisphere (if included in the network) have approximately the same rank among each of the two hemispheres. 
``L'' represents the left hemisphere, and ``R'' represents the right hemisphere.
The non-zero elements in $\hat{P}$ are mostly located in the two off-diagonal blocks, indicating that most nodes in the two hemispheres are exchanged.
Moreover, most non-zero elements in $\hat{P}$ in the two off-diagonal blocks are located on the diagonals of the two blocks, indicating that not only are most nodes permuted to the opposite hemisphere, they are permuted to their counterparts in the opposite hemisphere.
In order to interpret the $E=4,978$ that we found for the C. elegans neural network, and to compare this value to what we would expect from a similar network with no apparent symmetric structure, we did degree-preserving randomization to the C. elegans neural network and measured its $E$, preserving the total sum of weights of all edges in the network (which makes the $E$ of the two networks comparable), as well as the weighted degrees of nodes (which excludes the effect of node degree on network symmetry in the comparison).
In Figure~\ref{fig:C_elegans}(d) we show the results of $E$ for the randomized C. elegans neural network, $20\%$ higher than that of the original network, showing evidence that the C. elegans neural network sculpted by evolution has a certain level of symmetry.

\begin{figure*}
    \centering
    \includegraphics[width=\textwidth]{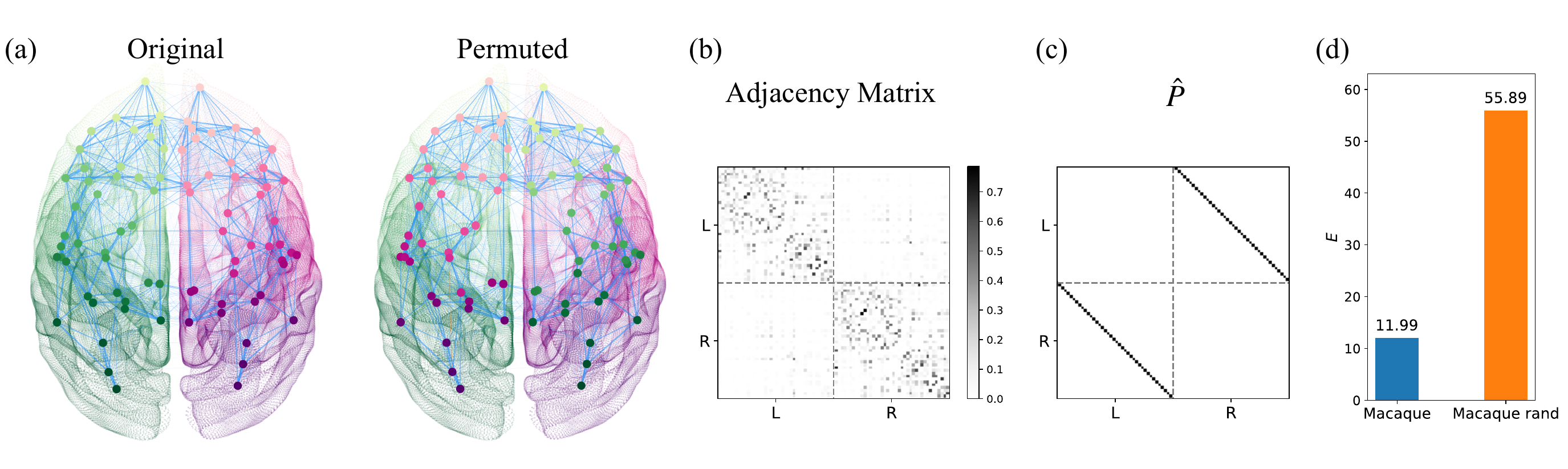}
    \caption{\scriptsize Symmetry of the Macaque brain network \cite{shen2019macaque}. 
    (a): The original Macaque brain network \cite{shen2019macaque} and the network after node permutation. In the original network the nodes are located at the centers of the brain regions that they represent (top view). The nodes are colored to distinguish the two hemispheres.
    (b): The adjacency matrix of the network \cite{shen2019macaque}. The network is weighted. 
    (c): The optimal node permutation matrix $\hat{P}$. In the adjacency matrix in (b) and $\hat{P}$ the rows and columns are ordered according to the hemispheres that the corresponding nodes belong to, and each node and its corresponding node in the other hemisphere have the same rank among their hemispheres. ``L'' represents the left hemisphere, and ``R'' represents the right hemisphere.
    (d): $E$ of the C. elegans brain network and the network after a degree-preserving randomization. 
    }
    \label{fig:macaque_before_after_perm}
\end{figure*}

\subsection{Macaque brain network}

Compared to C. elegans, Macaque is a much more advanced and complicated species, and has a much more complicated brain system \cite{kaiser2011evolution}. 
In \cite{shen2019macaque} the Macaque brain network was extracted, with brain regions as nodes ($N=82$), and connections between brain regions as edges ($M=3,312$) (Figure~\ref{fig:macaque_before_after_perm}(b)). 
In Figure~\ref{fig:macaque_before_after_perm}(a) we visualize the Macaque brain network (left). The nodes and brain regions in the left and right hemispheres are colored in different color tones in order to distinguish them. 
% We applied our algorithm to this network which found $\hat{P}$ that yields $E=11.99$. 
The optimal node permutation matrix $\hat{P}$ of the Macaque brain network, which represents the symmetry of the network, is shown in Figure~\ref{fig:macaque_before_after_perm}(c), and the permuted network is shown in Figure~\ref{fig:macaque_before_after_perm}(a) (right).
In the adjacency matrix of the network (Figure~\ref{fig:macaque_before_after_perm}(b)) and $\hat{P}$ (Figure~\ref{fig:macaque_before_after_perm}(c)), the rows and columns are ordered according to the hemispheres that the corresponding nodes belong to, and each node and its counterpart node in the opposite hemisphere have the same rank among their hemispheres.
The non-zero elements in $\hat{P}$ are located on the diagonals of the two off-diagonal blocks (Figure~\ref{fig:macaque_before_after_perm}(c)), indicating that every node is permuted to its exact counterpart in the opposite hemisphere, effectively switching the left and right hemispheres (Figure~\ref{fig:macaque_before_after_perm}(a)). 
% This is reflected by the fact that the non-zero elements in $\hat{P}$ are located on the diagonals of the two off-diagonal blocks (Figure~\ref{fig:macaque_before_after_perm}(c)). 
However, despite the nodes being permuted to their counterparts in the opposite hemisphere exactly, $E$ of the Macaque brain network is $11.99$ instead of zero.
This is due to the small differences in the wiring of nodes in the two hemispheres, which means that the Macaque brain network does not have perfect bilateral symmetry. 
Nonetheless, $\hat{P}$ as well as a low value of $E$ are able to capture the bilateral symmetry of the Macaque brain network, despite the small imperfection of the symmetry, which would have been undetected if only perfect symmetries are considered. 
We did degree-preserving randomization to the network and measured its $E$, which is almost five folds of that of the original network (Figure~\ref{fig:macaque_before_after_perm}(d)), confirming the high level of structural symmetry of the macaque brain region network. 
% Even though in the node permutation the left and right hemispheres have been switched perfectly, $E$ is still non-zero.
% This is because despite the high level of similarity in the left and right hemispheres of the macaque brain, they are not completely the same, which means that the macaque brain region network does not have a perfect symmetry. 
% Our algorithm, however, is able to capture the near-perfect approximate symmetry of the system, which would have been undetected if only perfect symmetries are considered.

The level of bilateral symmetry is higher in the Macaque brain network than that of the C. elegans neural network, and 
the difference between $E$ of the real and randomized macaque brain region network in proportion is greater than that of the C. elegans neural network.
We speculate that one of the contributing factors behind this is that the C. elegans brain network is built on the neuron level, while the Macaque network is built on the brain region level. 
Because the level of noise and randomness is higher at the neuron level compared to the more integrated brain level region, the C. elegans brain network has a lower level of symmetry.
% Therefore, the higher level of noise and randomness decreased the level of symmetry in the C. elegans neural network. 

\subsection{U.S. Senate Networks}

The U.S. Senate network is a temporal network describing the relations between all Senates at different times in history \cite{neal2020sign,neal2014backbone}. The nodes in this network represents U.S. Senates, and the edges are signed, representing positive or negative political relations, extracted from bill co-sponsorship data \cite{neal2020sign}. 

We measured the symmetry of the U.S. Senate network at different time periods, and show in Figure~\ref{fig:US_Senate}(a) $E$ of the network changing through time from 1973 to 2015. 
We can see that $E$ fluctuates while decreasing as a general trend with time, indicating that the U.S. Senate network becomes more symmetric over time.
To better understand the reason behind this pattern, and to interpret the symmetry in the network at different time periods, in Figure~\ref{fig:US_Senate}(b) we show six snapshots of the U.S. Senate network, and the symmetry found in each snapshot. 
Nodes that represent Senates affiliated with the Republican party are colored in red, affiliated with the Democratic party are colored in blue, and affiliated with independent parties are colored in green. 
In both the original network snapshots and the permuted network snapshots, the edges that are preserved are colored in blue, and the ones that are not preserved are colored in orange. 
In the permutation matrices $\hat{P}$, the rows and columns are separated into `R', `D' and `I' sections, representing the Republican, Democratic and independent affiliated Senates. The non-zero elements in $\hat{P}$ located in the diagonal R-R, D-D or I-I blocks represent within-party symmetry, and the non-zero elements located in off-diagonal areas represent cross-party symmetry.
We can see in Figure~\ref{fig:US_Senate}(b) that in 1973 and 1981, nodes belonging to the two main parties are mixed, and the preserved edges are distributed among within-party edges as well as between-party edges. Starting from 1989, the network becomes more and more segregated by parties, and a larger portion of the preserved edges are within-party edges. The sub-networks formed by nodes within parties become more symmetric, which is reflected by $\hat{P}$ -- there are almost no node-exchange between different parties.
Interestingly, although the network snapshots of years 1997-2013 are very modular with clear module division and comparable module sizes, they have very different symmetry compared to that of SBM networks (Sec. \ref{sec:SBM}, Figure~\ref{fig:SBM_Q2_Q3_examples}(c,d)). 
As explained in Sec. \ref{sec:SBM}, the symmetry between different modules and the symmetry within each module act as two competing contributors to the total symmetry of an SBM network, and a strong symmetry between different modules results in exchanges between nodes in different modules in $\hat{\mathbbm{P}}$ (Figure~\ref{fig:SBM_Q2_Q3_examples}(c,d)). However, the same type of symmetry is not observed in the U.S. Senate network. We speculate that this is because the symmetry of nodes in each party among themselves is stronger than the symmetry between nodes in the two main parties. Therefore the best symmetry permutation $\hat{\mathbbm{P}}$ of the U.S. Senate network is the combination of node permutations within each party.

\begin{figure*}
    \centering
    \includegraphics[width=\textwidth]{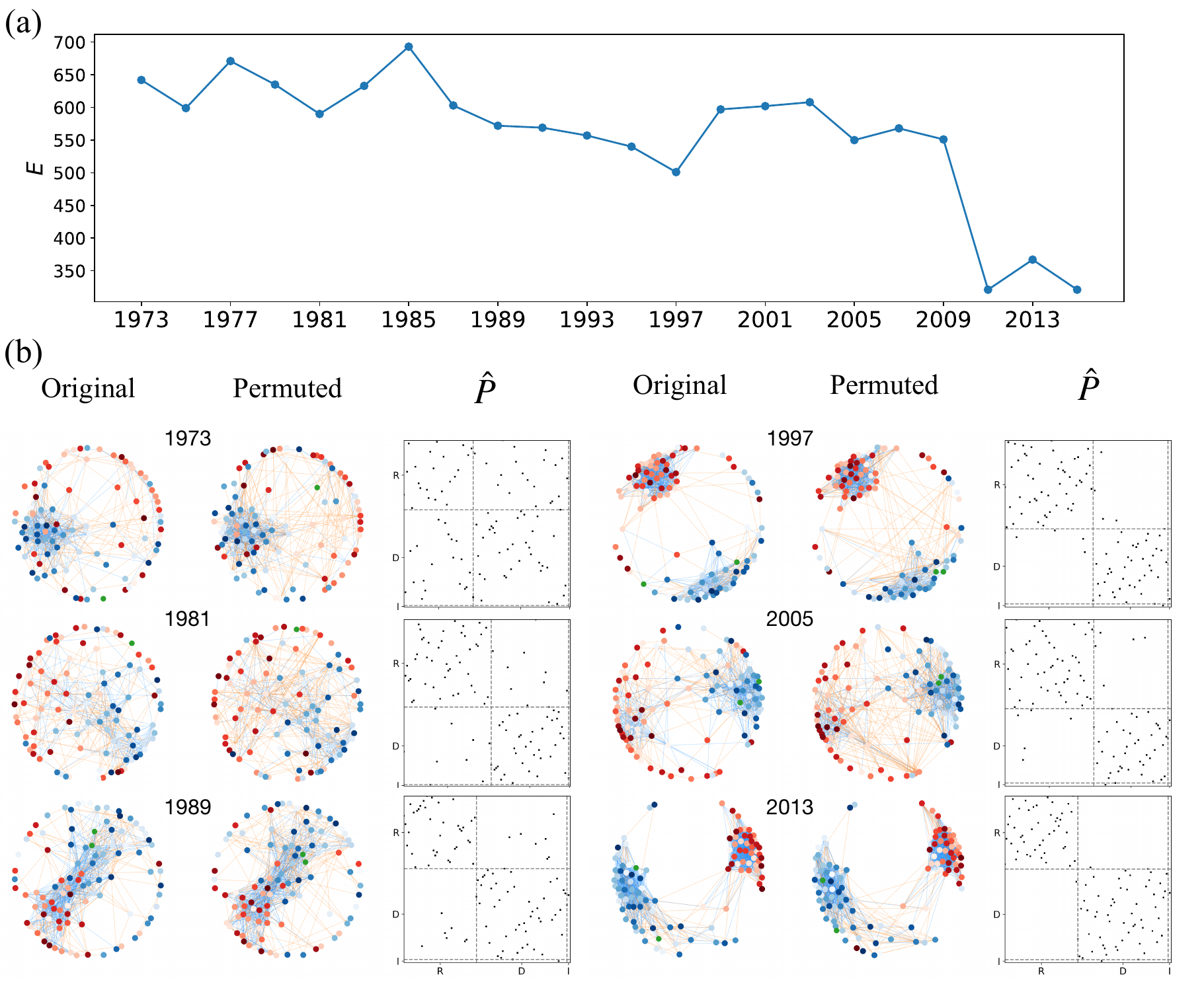}
    \caption{\scriptsize Symmetry of the U.S. Senate network \cite{neal2020sign,neal2014backbone}. 
    (a): $E$ of the U.S. Senate network \cite{neal2020sign,neal2014backbone} at different snapshots over time. 
    (b): Six snapshots of the original U.S. Senate network \cite{neal2020sign,neal2014backbone} and the snapshots after node permutations, as well as the optimal node permutation matrices $\hat{P}$. In the network, the nodes that represent Republican Senates are colored in red, the nodes that represent Democrat Senates are colored in blue, and the nodes that represent Senate in independent parties are colored in green. 
    The rows and columns of $\hat{P}$ are ordered according to the party affiliations of the corresponding nodes. ``R'' represents the Republican Party, and ``D'' represents the Democrat Party, and ``I'' represents independent parties.
    }
    \label{fig:US_Senate}
\end{figure*}

% \section{spectrum}

\section{Conclusion}

In this paper we explored the global symmetries of networks. We incorporated the concept of approximate symmetry into the definition of network symmetry, in contrast to other measures of symmetry in networks which relate to exact symmetries. We explored several classical network models and evaluated the level of symmetry of networks generated by these models with various different parameter values, and discovered useful insights concerning network structures. 

Specifically, we found that among the network models that we have explored, ER networks have the lowest level of symmetry compared with networks generated from other models with the same numbers of nodes and edges. We also found that ER networks sampled from ensemble $G(N,M)$ have the same expected level of symmetry as networks sampled from $G(N,N(N-1)/2-M)$, which can be explained by the fact that each network has the same level of symmetry as its network complement. 

For RGGs, the existence of periodic boundary conditions decreases the level of symmetry, which is counter-intuitive, because we would expect that the existence of periodic boundary conditions creates an additional translational symmetry on top of the symmetries of RGGs without periodic boundary conditions. 
We also found that the symmetry of RGGs most commonly comes from the high similarities between nodes that are close together in space. 
% When there are periodic boundary conditions, nodes are more equivalent to one another than when there are not periodic boundary conditions. This fact is explain by the fact that symmetries in RGGs mostly come from similarities between nodes that are close together in space, so that if there are periodic boundary conditions, then, ..., I don't know why.

For Watts-Strogatz networks, as the rewiring probability increases, the level of symmetry decreases from a perfect symmetry to the level of symmetry expected from an ER network with the same number of nodes and edges, and saturates for all levels of rewiring probability beyond that.

For SBM networks, we discovered that the existence of high-density modules increases the level of symmetry relative to an ER network with the same numbers of nodes and edges, because of the similarities between different modules. For example if we have two same-sized modules with the same density, then the optimal permutation of them includes the interchange of nodes between these two modules.

Finally, we analyzed three examples of real-world networks. Both the C. elegans brain network and the Macaque brain network have a certain level of bilateral symmetry, especially for the Macaque brain network whose optimal node permutation is interchanging nodes in each hemisphere with their counterparts in the opposite hemisphere with precision.
The temporal U.S. Senate network becomes more symmetric over time, caused by the gradual segregation between the Senates affiliated to the two main political parties.

The approximate network symmetry that we defined and measured in this paper provides another dimension to the network models we so often study, in addition to all of the other network measures that we typically use. 
This can also be very useful in dealing with real-world networks: first evaluate whether a real-world network has a certain level of symmetry, and second find out exactly what the symmetry is composed of -- for example whether it is bilateral symmetry, symmetry based on spatial embedding, or symmetry resulted from interchanging highly connected communities and so on. 
Here we only used three real-world network examples, but this method can be applied to any real-world network.
% For example, network symmetry has been studied in gene regulatory networks [A BUNCH OF CITATIONS] and metabolic networks [CITATIONS].

One possible aspect of future work is exploring the symmetry of other real-world networks as well as networks generated by other network models. 
% Another possible aspect would be to improve the algorithm used to obtain optimal permutations of nodes. The current algorithm obtains permutations that could very well not be the global optimum, but the result is definitely correlated with the real global optimum such that the results and analysis on different network models and real networks are still meaningful. 
Another possible aspect is exploring the relations between network symmetry and other properties of networks and network ensembles. 
Moreover, it would be also interesting to look into the relation between approximate network symmetry and network spectrum.
% For example, we speculate that there might be a relation between the expected network symmetry and the y. 

% \begin{acknowledgments}
% % put your acknowledgments here.
% \end{acknowledgments}

\begin{appendices}
% \m{appendix}
\section{Algorithm}
% \appendix{Algorithm}

We use an Markov-Chain Monte-Carlos (MCMC) algorithm that minimizes $\mathcal{E}(A, P)$ for a given adjacency matrix $A$ over possible global node permutation matrix $P$. 
We use a logarithmic cooling scheme in the MCMC algorithm: $T(t) = \frac{c}{\log(t+d)}$, where $T(t)$ is the temperature used in the acceptance probability at time-step $t$, and $c$, $d$ are parameters. 

We do not claim that the our current algorithm is always able to obtain the absolute global optimal permutations of all networks, nor do we claim that it is the most efficient algorithm in finding the optimal node permutations.
However, we believe the optimal permutations obtained by our algorithm, even though might not be the absolute global optimal, is correlated with the global optimal such that the results and analysis on different network models and real-world networks based on the obtained optimal permutations are representative of that of the absolute optimal permutations. 

\section{Distribution of $\mathcal{E}(A, P)$}

\begin{figure}
    \centering
    \includegraphics[width=0.45\textwidth]{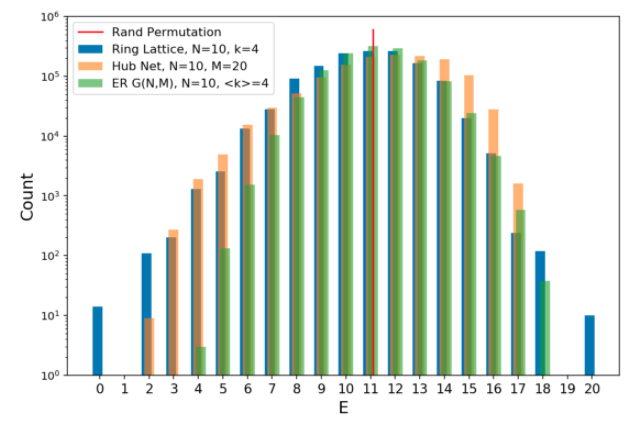}
    \caption{\scriptsize
    The distributions of $\mathcal{E}(A, P)$ for three networks: a ring lattice with $N=10$ and node degree $k=4$, a hub network with $N=10$ and $M=20$, and an ER network with $N=10$ and $M=20$. 
    The red line marks the expected $\mathcal{E}$ from a uniformly randomly chosen node permutation (without the global permutation restriction), for a network with $N=10$ and $M=20$, whose edges are uncorrelated with one another.  
    }
    \label{fig:spectrum}
\end{figure}

\begin{figure}
    \centering
    \includegraphics[width=0.47\textwidth]{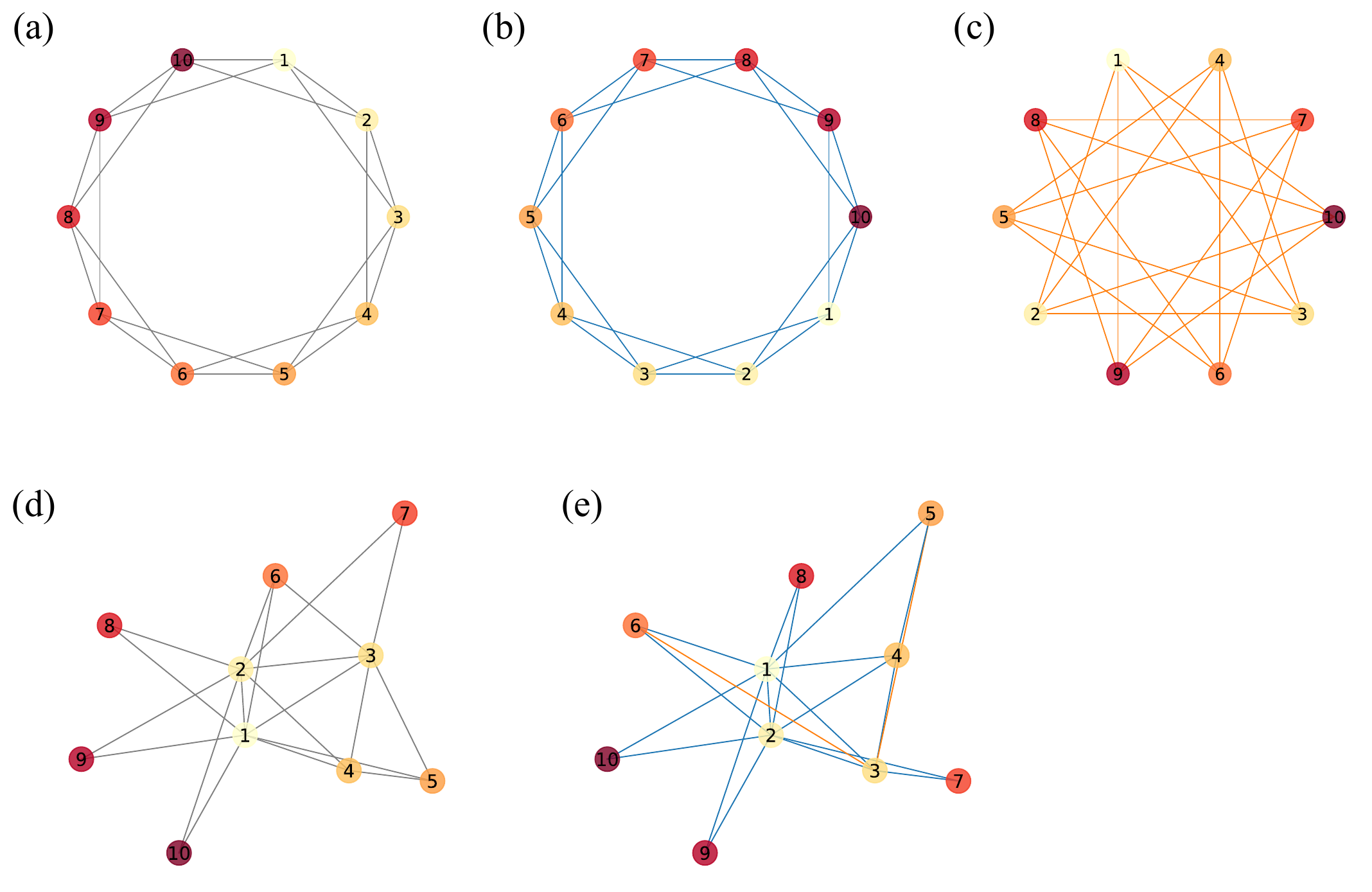}
    \caption{\scriptsize
    (a): A ring lattice network with $N=10$ and node degree $k=4$
    (b): The ring lattice network in (a) after a node permutation that gives rise to $\mathcal{E}(A, P)=E(A)=0$.
    (c): The ring lattice network in (a) after a node permutation that gives rise to $\mathcal{E}(A, P)=\mathcal{E}_{\max}(A, P)=20$.
    (d): A hub network with $N=10$ and $M=20$.
    (e): the hub network in (d) after a node permutation that gives rise to $\mathcal{E}(A, P)=\mathcal{E}_{\min}(A, P)=0$.
    }
    \label{fig:small_net_examples}
\end{figure}

In this section we briefly explore the distribution of all possible values of $\mathcal{E}(A, P)$ for a given network with adjacency matrix $A$. 
In order to obtain a distribution of $\mathcal{E}(A, P)$ for a network, we need to find all possible global node permutations $\mathbbm{P}$ for the network.
% We define $\mathcal{E} = \frac{1}{4} (|| A - QAQ^T ||_1)$, which is an extention of $E$ without the minimization.
For any network, the minimum value of $\mathcal{E}(A, P)$ is $E(A)$, and the maximum value of $\mathcal{E}$ cannot exceed neither $M$ nor ${N\choose 2} - M$ ($\mathcal{E}_{\max} \leq \min(M, {N\choose 2} - M)$). 
The number of $P$ for a network with $N$ nodes scales as $N!$, which greatly limits out ability to find all $P$ for a large $N$.
Here we use three networks all with $N=10$ and $M=20$ as examples to study the distributions of $\mathcal{E}(A, P)$, shown in Figure~\ref{fig:spectrum}. 
The red line in Figure~\ref{fig:spectrum} marks $\mathcal{E}_{rand}=M \left[ 1 - \frac{2M}{N(N-1)} \right]$, which represents the expected number of edges that are preserved in a random node permutation (without the global permutation restriction), with the assumption that the edges are not correlated. 
$\mathcal{E}_{rand}$ is calculated as follows: assume that there is a string of zeros and ones, with in total $\frac{N(N-1)}{2}$ digits, and $M$ ones. After a random shuffle of the digits, the expected number of overlapping ones between the original string and the string after permutation is $M \left[ 1 - \frac{2M}{N(N-1)} \right]$.
The first network is an ring lattice network (Figure~\ref{fig:small_net_examples}(a)), whose all possible $\mathcal{E}$ are shown in Figure~\ref{fig:spectrum} in the blue histogram. The minimum $\mathcal{E}$ for the ring lattice is zero, which means it has $E=0$, indicating a perfect symmetry (Figure~\ref{fig:small_net_examples}(b)). 
The maximum $\mathcal{E}$ is $20$, which is equal to the number of edges, meaning that none of the edges are preserved in this node permutation (Figure~\ref{fig:small_net_examples}(c)).  
We can see that in Figure~\ref{fig:spectrum} $\mathcal{E}_{rand}$ coincides with the peak of the distribution of $\mathcal{E}$ of the ring lattice, suggesting that the approximation used in the calculation of $\mathcal{E}_{rand}$ is acceptable in the case of the ring lattice network.
The second network is an ER network sampled from a $G(N,M)$ ensemble with $N=10$ and $M=20$, and its $\mathcal{E}$ are shown as the green histogram in Figure~\ref{fig:spectrum}. 
We can see that for this ER network, $\mathcal{E}_{\min}$ is no longer zero, which means that it does not have perfect symmetry. The peak of the histogram also coincides with $\mathcal{E}_{rand}$, suggesting that the approximation used in the calculation of $\mathcal{E}_{rand}$ is acceptable in the case of this ER network too. 
Notice that the histogram is less wide-spread compared to that of the ring lattice.
The third network is a hub network with heterogeneous node degrees, with the same $N$ and $M$ (Figure~\ref{fig:small_net_examples}(d)). Compared to the ER network, the hub network has a lower $\mathcal{E}_{\min}$, which is caused by the existence of hubs (the network after the optimal node permutation is shown in Figure~\ref{fig:small_net_examples}(e)). 
For example, in the extreme case of a heterogeneous-degree network, a star network, there is a high level of symmetry due to the equivalence between the leaf nodes. 
Compared with the histograms of the ring lattice and the ER network, the histogram of $\mathcal{E}$ of the hub network is slightly skewed towards higher values of $\mathcal{E}$. This can be explained by the heterogeneous nature of the node degrees - when a hub node and non-hub node exchange places, it could create a large number of non-preserved edges.

% \section{Some More Notation}
\end{appendices}

\clearpage

\bibliographystyle{plain}
\bibliography{dlarva-200518}

\end{document}